\documentclass[a4paper,12pt,reqno,nofootinbib]{revtex4}

\usepackage[centertags]{amsmath}
\usepackage{amsfonts}
\usepackage{amssymb}
\usepackage{amsthm}
\usepackage{newlfont}
\usepackage{stmaryrd}
\usepackage{mathrsfs}
\usepackage{euscript}
\usepackage{graphicx}
\usepackage{enumerate}
\usepackage{tikz}
\usepackage{pgf}
\usetikzlibrary{positioning,fit,calc}
\usetikzlibrary{arrows,automata}
\usepackage{wrapfig}
\usepackage{subfigure}
\usepackage{amscd}
\usepackage{hyperref}

\theoremstyle{plain}
  \newtheorem{theorem}{Theorem}[section]
  
  \newtheorem{proposition}[theorem]{Proposition}
  
\theoremstyle{definition}

\theoremstyle{remark}

\newcommand{\opunit}{\text{1}\kern-0.22em\text{l}}

\DeclareMathAlphabet{\mathpzc}{OT1}{pzc}{m}{it}

\newcommand{\id}{\textrm{d}}

\begin{document}

\title{{\bf What decides the direction of a current? }}

\author{Christian Maes}
\affiliation{Instituut voor Theoretische Fysica, KU Leuven}

\begin{abstract}
Nonequilibria show currents that are maintained as the result of a steady driving. We ask here what decides their direction.  It is not only the second law, or the positivity of the entropy production that decides; also non-dissipative aspects often matter and sometimes completely decide.
\end{abstract}

\maketitle

\begin{center}
{\bf Dedicated in honor of Lucio Russo}
\end{center}
\baselineskip=20pt
\section{Introduction}
Predicting the course of events given the present state of affairs is part of scientific practice.  In what direction things will evolve is however not always so evident.  In thermodynamics there are a number of general rules of thumb derived from the principal laws.  For instance, macroscopic systems tend to equilibrate at the same temperature, chemical potential and pressure as the surroundings; relaxation (or time itself) flows in the direction of increasing entropy etc. In mechanics the ambition is even bigger; we compute trajectories given the present state. Statistical mechanics is supposed to transfer mechanical laws to thermodynamic behavior, with the attenuendo that some thermodynamic principles are not absolute but become statistical.  For example, the Boltzmann equation for a dilute gas has a direction of time, but for mesoscopic systems fluctuations can be expected, and as Maxwell was emphasizing, 
\begin{quote}
The truth of the second law is ... a statistical, not a mathematical, truth, for it depends on the fact that the bodies we deal with consist of millions of molecules... Hence the second law of thermodynamics is continually being violated, and that to a considerable extent, in any sufficiently small group of molecules belonging to a real body.
\end{quote}
(J.C. Maxwell, 1878)\\
That is, statistical mechanics will not only derive thermodynamics, it will also correct it and extend it.  That is especially true for nonequilibrium statistical mechanics as we are dealing there necessarily with un-typicial behavior from the point of view of the micro-canonical ensemble.  It becomes therefore both a major inspiration and application of probability theory, exactly in the way Lucio Russo has been enjoying it and contributing to it.\\
Going to irreversible thermodynamics \cite{Groot}, that is the thermodynamics for irreversible phenomena, the main guiding principle that survives for the direction of currents is the positivity of the entropy production.  We are for example considering an open macroscopic system which is being steadily frustrated by contacts with different equilibrium baths.  There will be currents maintained, at least on the time scales where the environment be kept at the same intensive values (e.g. temperature).  The directions of these currents can and will vary with different arrangements, but the entropy production $\Sigma$  is positive. That $\Sigma =\sum_\alpha J_\alpha F_\alpha$ is a sum over all possible types  of channels of transport of the product of currents (or displacements) $J_\alpha$ and thermodynamic forces $F_\alpha$.  For predicting the current directions, we just see what is compatible with $\Sigma\geq 0$, nothing more.  In the linear regime, where currents are proportional to forces, $J_\alpha=\sum_\gamma L_{\alpha\gamma}\,F_\gamma$ with symmetric\footnote{We ignore here the Casimir correction that takes into account the parity under time-reversal of the physical quantity being transported.} Onsager linear response coefficients $L_{\alpha\gamma}=L_{\gamma\alpha}$, and the positivity of $\Sigma$ is the positivity of the matrix $\left(L_{\alpha\gamma}\right)$.  Here again, statistical mechanics will derive and extend that scheme, but now it should be a nonequilibrium statistical mechanics.  That is very much unfinished business, and certainly for going beyond the linear regime around zero thermodynamic forces.
In fact, nonequilibrium statistical mechanics is far behind its equilibrium version, 
\begin{quote}
My inclination is to postpone the study of the large--system limit: Since it is feasible to
analyze the nonequilibrium properties of finite systems --- as Gibbs did for their equilibrium
properties --- it seems a good idea to start there. That may not answer all questions, but
it advances nonequilibrium statistical mechanics to the point equilibrium had reached after
Gibbs.
\end{quote}
{\it Conversations on Nonequilibrium Physics With an Extraterrestrial}, David Ruelle, Physics Today {\bf 57}(5),
48 (2004).\\
In other words, a general theory of nonequilibrium phase transitions or of universality is still out, and
even a systematic way of dealing with many-body effects is largely lacking.  We have certainly no percolation or geometric picture
of nonequilibrium collective phenomena, and remembering the crystal clear and perfectly elegant contributions of Lucio to percolation theory and to mathematical statistical mechanics, we can only hope that the day will soon come where such a mathematical framework and geometric interpretation will become available also for nonequilibrium physics to match Lucio's standards.\\

In what follows we are asking about what determines the direction of a nonequilibrium current.  The main point will be that it is certainly not always the case that the current direction is decided by the positivity of the entropy production; non-dissipative effects will be important and sometimes crucial.  We refer to the pedagogical introduction \cite{nond} on non-dissipative aspects of nonequilibrium statistical mechanics.  For the moment it suffices to add that transition rates in a process also have time-symmetric parameters and, quite obviously, we need to understand how they contribute to deciding the direction of the current.

\tableofcontents
\section{Traditional arguments}

\subsection{Phenomenology}\label{phen}
The medium inside and outside of a biological cell can be very different. These are connected via thin pores through which ions of various chemicals can be transported.  Consider such a pore or channel in the membrane separating outside and inside; see Fig.~\ref{memb}.  Because of different concentrations at its ends, there will be a current through the pore.  In fact, ions will be travelling from the region of higher chemical potential to the region of lower chemical potential.  The same thing happens with many types of currents, whether the channel is connected to thermal, chemical or mechanical reservoirs.
\begin{figure}[h]
\centering
\includegraphics[width=5 cm]{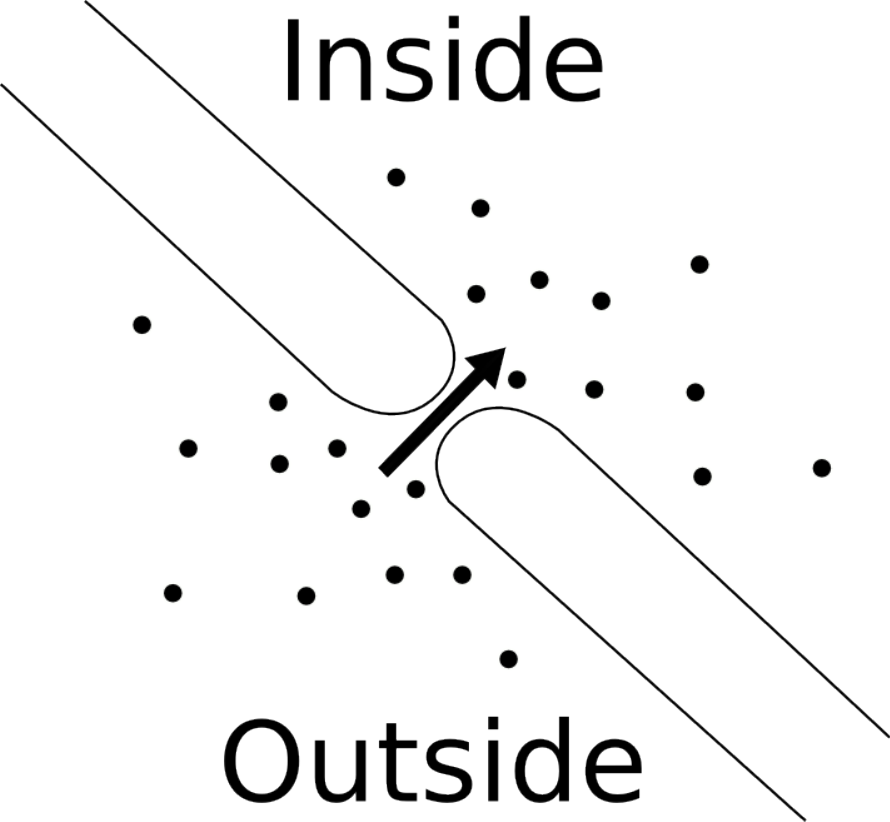}
\caption{Example of a simple stationary current for which the direction is decided by the positivity of the entropy production.}\label{memb}
\end{figure}
On the appropriate scale of time, the system is in steady nonequilibrium, not changing its macroscopic appearance.  There is a constant production $\Sigma$ of entropy in the environment, which is positive, \[
 \Sigma = -\beta\mu_1 \,J_1 - \beta \mu_2\,J_2
 \]
  with $J_i$ the particle flux into the $i-$th reservoir at chemical potential $\mu_i$ and inverse temperature $\beta$.  Stationarity (and bulk conservation of particles) implies that $J_1 + J_2 =0$ so that we can find the direction of the particle current $J_1$ by requiring
  \[
  \Sigma = \beta J_1(\mu_2 - \mu_1)\geq 0\quad \text{ (second law)}\]
By bulk conservation of particles we still have $J_1=J$ the stationary particle current through the channel or pore, from the second towards the first reservoir, and hence $J\geq 0$ whenever $\mu_2\geq \mu_1$.\\
Similar scenario's can be written for thermal and mechanical baths that frustrate the system.  Those are the typical cases where finding the direction of the current amounts to applying the second law in the form that the stationary entropy production be positive.\\
 
While the previous case was treated rather phenomenologically, precise mathematical arguments  can be provided for simple particle model systems following the same physics.  Here comes an example.

 \subsection{Stochastic lattice gas}
 We consider identical
 particles that can jump from site $i$ to nearest neighbor site $j=i\pm 1$ on the
 finite linear chain $\Lambda_N =
 \{-N,-N+1,\ldots,0,1,\ldots,N-1,N\}$; see FIg.~\ref{exx}.  The endpoints $i=\pm N$ in
 $\Lambda_N$ are called
 the boundary of the system; the other sites are in the bulk.
There is
 at most one particle per site $i$, so that a site $i$ can be vacant or
 occupied, and we write $\eta(i)\in \{0,1\}$ for
  the occupation at site $i\in \Lambda_N$. The state space is  $K = \{0,1\}^{\Lambda_N}$ with elements $\eta,\eta',\xi,\ldots \in K$. The reasoning below is outlined in \cite{prag}.\\
 
 The energy
 function on $K$ is chosen as
 \begin{equation}\label{hamilt}%
 H(\eta)=-B\sum_{i=-N}^{N} \eta(i) - \kappa \sum_{i=-N}
 ^{N-1} \eta(i)\,\eta(i+1),%
 \end{equation}%
 where $B$ and $\kappa$ are some real constants. 
 The system is imagined in thermal contact with a very
 large heat bath at inverse temperature $\beta$ (Boltzmann's
 constant is set equal to one).
The energy change in that bath over the transition $\eta\rightarrow \eta'$  gives a first contribution $\beta (H(\eta) -
  H(\eta'))$ to the change of
 entropy in the reservoir.
  Another important quantity here is the particle number,
 \[
 {\cal N}_{[j,k]}(\eta) =\sum_{i=j}^k\eta(i)
 \]
in the lattice interval
 $[j,k]\cap\Lambda_N$, $-N\leq j\leq k\leq N$. The total number of particles
 inside the system is $\cal N = {\cal N}_{[-N,N]}$.\\
 We imagine now also that the system is in contact with a
 particle reservoir at each of its boundary sites. There can be a
 birth or a death of a particle at these sites, which amounts to the entrance from and the exit to the corresponding reservoir of a particle.
 In that sense we write $J_\ell =
   \Delta {\cal N}_\ell$, $J_r = \Delta {\cal N}_r$ as the changes
   in particle number in the {\it left}, respectively {\it right}
   particle reservoir. The flow of particles in and out of the system
 can also contribute to the dissipated heat in the reservoir, and
 hence to changes in entropy:
 \begin{equation}\label{sent}
  S(\eta,\eta') = \beta\Delta E(\eta,\eta') - \beta\mu_\ell \Delta {\cal
  N}_\ell(\eta,\eta') - \beta\mu_r \Delta {\cal N}_r(\eta,\eta')
  \end{equation}
  is
   the change of entropy in the environment
 for $\mu_\ell$, respectively $\mu_r$ the
  chemical potentials (up to some factor $\beta$ that we have
  ignored) of the particle reservoirs left and right.  We will make mathematical sense of
 \eqref{sent} in terms of variables inside the system.
 \begin{figure}[h]
 \centering
 \includegraphics[width=10 cm]{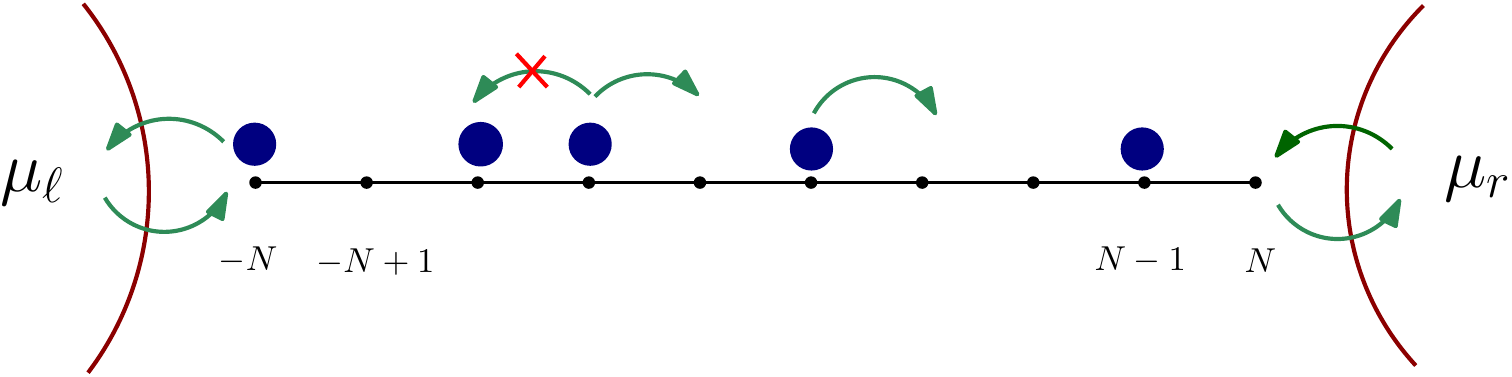}
 \caption{Stochastic lattice gas, symmetric in the bulk and governed by local interactions, driven by contact with particle reservoirs at different chemical potentials.}\label{exx}
 \end{figure}

 For the dynamics we choose a continuous time Markov process on $K$. Write the
 transformation
\[ 
 \eta ^{i,j} (k) =
 \begin{cases}
 \eta (k) & \text{if } k\neq i,\, k\neq j;
 \\
 \eta (i) & \text{if } k= j;
 \\
 \eta (j) & \text{if } k= i
 \end{cases}
\] 
for the state obtained from $\eta$ after exchanging the occupation of the sites $i,j$, only allowed for $j = i \pm 1$. The rate
 for that transition is taken to be
 \begin{equation}\label{cij}%
 C(i,j,\eta)=\exp \Bigl[ -\frac{\beta}{2} (H(\eta^{i,j})
 -H(\eta))\Bigr]\,,\quad |i-j|=1%
 \end{equation}%
Similarly, the rate of birth and death for the transition
 $\eta \rightarrow
 \eta ^{i}$ with
 \[ 
  \eta^{i}(k) =
  \begin{cases}
  1- \eta(k) & \text{if } k=i
  \\
  \eta(k) & \text{if } k\neq i
  \end{cases}
 \] 
 only occurring at sites $i=-N,N$, is 
 \begin{equation}\label{ci}%
 C(i,\eta)=e^{-a_{i}\eta(i)}\exp\Bigl[ -\frac{\beta}{2}
 (H(\eta^{i}) - H(\eta))\Bigr]
 \end{equation}%
The relevant parameters are the values $a_{-N}=\beta\mu_\ell, a_{N} =\beta\mu_r$ representing the (different) chemical potentials of the
 two reservoirs at the outer edges.\\

 One observes from the definitions (\ref{cij}) that:
 \begin{equation}\label{detail1}
     \frac{C(i,j,\eta)}{C(i,j,\eta^{i,j})}=
     \frac{\exp \left [ -\beta H\left (
 \eta^{i,j}\right ) \right ]}{\exp \left [ -\beta H\left (
 \eta\right ) \right ]}
 \end{equation}  Furthermore, from (\ref{ci}) we have
 \begin{equation}\label{detail2start}%
 \frac{C(i,\eta)}{C(i,\eta^{i})}= \frac{\exp [-a_{i} \eta(i)]}{\exp
 [-a_{i} (1-\eta(i))]}\frac{\exp \left [ -\beta H\left (
 \eta^{i}\right ) \right ]}{\exp \left [ -\beta H\left ( \eta
 \right ) \right ]}\,,\quad i=\pm N%
 \end{equation}%
 For
 $a_{-N}=a_{N}=a$,
  when the particle reservoirs left and right have
 equal concentration, then the system dynamics satisfies the
 condition of detailed balance:   for all allowed transitions $\eta \rightarrow \eta'$ and
  corresponding transition rates $W(\eta\rightarrow \eta')$,
 \begin{equation}\label{detailbalance}%
 \frac{W(\eta\rightarrow \eta')}{W(\eta' \rightarrow
 \eta)}=\frac{\mathbb{P}^{\beta,a} [\eta']}{\mathbb{P}^{\beta,a}
 [\eta]}
 \end{equation}%
for the grand-canonical equilibrium probabilities
  \begin{equation} \label{gibbs}%
  \mathbb{P}^{\beta,a} [\eta]=\frac 1{{\cal Z}}\,e^{a\,\sum\eta(i)}\,e^{-\beta H(\eta)}%
  \end{equation}%
  where ${\cal Z}={\cal Z}(a,\beta,N)$ is a normalization factor.
 Thus, 
  \eqref{gibbs} is a reversible stationary measure when  $a_{-N}=a_{N}=a$.\\
  
   We now consider $a_1\neq a_N$ (different chemical potentials). At the left boundary of the system, see \eqref{detail2start},
 \begin{equation}%
 \label{aneqa1}%
 \frac{C(-N,\eta)}{C(-N,\eta^{-N})}=e^{-\beta
 [H(\eta^{-N})-H(\eta)]-a_{-N} J_{\ell} (\eta, \eta^{-N})}%
 \end{equation}%
 where $J_{\ell}(\eta, \eta^{-N})=1$ when the
  particle leaves the system via the site $-N$, i.e., $\eta(-N)=1$,
 and $J_{\ell}(\eta, \eta^{-N})=-1$ when a new particle enters,
 i.e., $\eta(-N)=0$.  Analogously, the current $J_{r}(\eta, \eta') = 1$ when
 $\eta(N)=1$, $\eta'=\eta^N$ and $J_{\ell}(\eta, \eta^{'})=-1$
 when $\eta(N)=0$, $\eta'=\eta^N$.  The currents are zero otherwise.\\
As a consequence,
 \begin{equation}\label{ldb}%
 \frac{W(\eta\rightarrow \eta')}{W(\eta' \rightarrow
 \eta))}=e^{-\beta[H(\eta')-H(\eta)]-a_{-N} J_{\ell}(\eta,
 \eta') - a_N J_r(\eta, \eta')}%
 \end{equation}%
 where we see the change of entropy \eqref{sent}. In other words,
 \begin{equation}\label{set}
 \frac{W(\eta\rightarrow \eta')}{W(\eta' \rightarrow \eta)}=
 e^{S(\eta,\eta')}
 \end{equation}
(which is known as the condition of {\it local} detailed balance),  and
 \begin{equation} \label{conserlow}%
 J_\ell (\eta, \eta')+J_r (\eta,\eta')={\cal N}(\eta)-{\cal N}(\eta')%
 \end{equation}%
 or,
with $a_N = a$, $a_{-N} = a + \delta$,
 \[
 \frac{W(\eta\rightarrow \eta')}{W(\eta' \rightarrow
 \eta)}=\frac{\mathbb{P}^{\beta,a} [\eta']}{\mathbb{P}^{\beta,a}
 [\eta]}\, e^{-\delta J_\ell(\eta,\eta')}
 \]
with $\delta$ thus measuring the amount of breaking of detailed balance.

 As above we define the bulk currents $J_i(\eta,\eta')$ to be $+1$
 if in the transition $\eta\rightarrow \eta'$ a particle moves over
 the bond $i\rightarrow i+1$, and equal to $-1$ if a particle moves
 $i \leftarrow i+1$.  In fact and throughout we confuse current
 with what is more like a time-integrated current, or a change of
 particle number.\\
 
  We have piecewise-constant paths $\omega$ over the
  time-interval $[0,\tau]$, starting from some initial configuration
  $\eta_0$ after which it changes into $\eta_{t_1},
  \eta_{t_2},\ldots$ at random times $t_1, t_2,\ldots$.  At the jump times we take
  $\eta_{t_{k-1}} = \eta_{t_k^-}$ and $\eta_{t_k} = \eta_{t_k^+}$ for having right-continuous paths with left limits.
 The time-reversal transformation on path-space  $\Theta$ is defined via $(\Theta\omega)_t  =
  \omega_{\tau-t}$, up to irrelevant modifications at the jump times
  making $\Theta\omega$ again right-continuous.\\
   We consider a path $\omega =
 (\eta_t)_{t=0}^\tau$ and currents $J_i(\omega)$, $i=-N,\ldots,N,$
 defined by \[ J_i(\omega) = J_i(\eta_0,\eta_{t_1}) +
 J_i(\eta_{t_1},\eta_{t_2}) + \ldots +
 J_i(\eta_{t_{n-1}},\eta_{\tau})
 \]
 In particular, $J_r = J_N$ and for $i\leq k$,
 \begin{align}\label{bc}
 J_{i} (\omega) - J_{k}  (\omega) &= {\cal
 N}_{[i+1,k]}(\eta_\tau)- {\cal N}_{[i+1,k]}(\eta_0)\nonumber\\
 J_\ell(\omega) + J_{-N}(\omega) &= \eta_0(-N) - \eta_\tau(-N)
 \end{align}%
 Observe that the currents $J_i$ are extensive in the time
 $\tau$.\\
 
 All of that is related to the process, be it transient or be it
 steady.  We concentrate on the steady state regime.  It is easy to
 verify that we have here a unique stationary distribution $\rho$.
  That stationary distribution is  only implicitly known, solution of the
  (time-independent) Master equation.  Corresponding to $\rho$ there is then a
 stationary process with distribution ${\boldsymbol P}_\rho$.  If
 we look at expectations in the stationary process we write
 $\langle \,\cdot\, \rangle_\rho$.\\
 From the conservation laws \eqref{conserlow} and \eqref{bc} we have
 \[
 \langle J_\ell \rangle_\rho = -\langle J_r \rangle_\rho = -
 \langle J_i \rangle_\rho, \qquad i \in \Lambda_N
 \]
 
 \begin{proposition}
 The direction of the current is from higher to lower chemical potential, i.e.,
 assuming that $\delta \geq 0$ (or, $a_{-N}=\mu_\ell \geq a_N =\mu_r$) we have $\langle
 J_i\rangle_\rho \geq 0$. 
 \end{proposition}
\begin{proof}
The path density of ${\boldsymbol P}_\rho$ with respect to
 ${\boldsymbol P}_\rho\Theta$, both started in the stationary
 distribution $\rho$, is
 \begin{equation}\label{concp}
 \frac{\id{\boldsymbol P}_{\rho }}{\id{\boldsymbol P}_{\rho }\Theta
 }(\omega )=\frac{\rho (\omega _{0})}{\rho (\omega _{\tau})}\exp
 \left[ -\beta \left(
 H(\omega _{\tau})-H(\omega _{0})\right) +a\Delta {\cal N}-\delta J_{\ell}(\omega )\right],%
 \end{equation}%
By normalization we
 have:%
 \[
 \int \id{\boldsymbol P}_{\rho} (\omega) \frac{\id{\boldsymbol
 P}_{\rho}\Theta}
 {\id{\boldsymbol P}_{\rho}}(\omega)=1.%
 \]
 and hence, by concavity,
 \begin{equation}\label{conca}
 \int \id{\boldsymbol P}_{\rho}
 (\omega) \log \frac{\id{\boldsymbol P}_{\rho}\Theta} {\id{\boldsymbol P}_{\rho}}(\omega)\leq 0.%
 \end{equation}
 But, from \eqref{concp} and by stationarity
 \[ 
 0 \leq \int \id{\boldsymbol P}_{\rho}\,\log \frac{\id{\boldsymbol
 P}_{\rho}} {\id{\boldsymbol
 P}_{\rho}\Theta}(\omega)=-\delta\langle J_\ell\rangle_\rho =
 \delta\langle J_i\rangle_\rho
 \] 
  We conclude that
 \[ 
 \delta \langle J_i\rangle_\rho\geq 0%
\] 
 which shows that the average direction of the particle current
 depends only on the sign of $\delta$.\\
 Getting a strict inequality $\langle J_i\rangle_\rho >0$ is also
 possible for $\delta >0$; it suffices to see that there is a
 non-zero probability that the current $J_i$ as a function of the
 path $\omega$ is not constant equal to zero even when
 $\omega_0=\omega_\tau$.
 \end{proof}

Looking back at the proof we see that the main inequality has been the positivity \eqref{conca} of the relative entropy between the forward and the backward stationary process.  The latter coincides with the stationary entropy production, as is in fact visible from  \eqref{set}.  Hence, the proof above, as in \cite{prag}, is a non-perturbative statistical mechanical argument or the physical analogue for the phenomenology in Section \ref{phen}; nothing really new here.
 
\section{Problematic cases}

We collect a number of situations where the previous either phenomenological or statistical mechanical arguments, based on the positivity of the entropy production, do not work.
From a general perspective comparable to the so called Curie principle, currents may appear whenever they are not forbidden by some symmetry.  It is then not wholly surprising that we cannot always apply the same physical arguments.  Yet, the examples below are specifically relevant in the context of nonequilibrium physics, for which we may hope to develop some framework. 
 
\subsection{Ratchet currents}
 
 \subsubsection{Triangula}\label{tra}

In \cite{cvdb} a number of  versions of hard disc microscopic ratchets are introduced
and studied with molecular dynamics and with some low density expansions. A directed systematic motion appears when a temperature difference is applied to different units of a motor. One of the simple examples is called there the Triangula: it is a motor consisting of two identical triangular units, each sitting in a gas (reservoir) consisting of hard discs whose centres collide elastically with the triangles; see Fig.~\ref{tir}.  The two triangles  are rigidly connected along a rod, with their base parallel to it, and the whole motor is constrained to
move along the horizontal direction without rotation or vertical displacement.  When the temperatures in the two reservoirs are different, there appears a systematic motion which turns out te be in the direction of the triangles's apices --- to the right in Fig.~\ref{tir}.
\begin{figure}[h]
\centering
\includegraphics[width=8 cm]{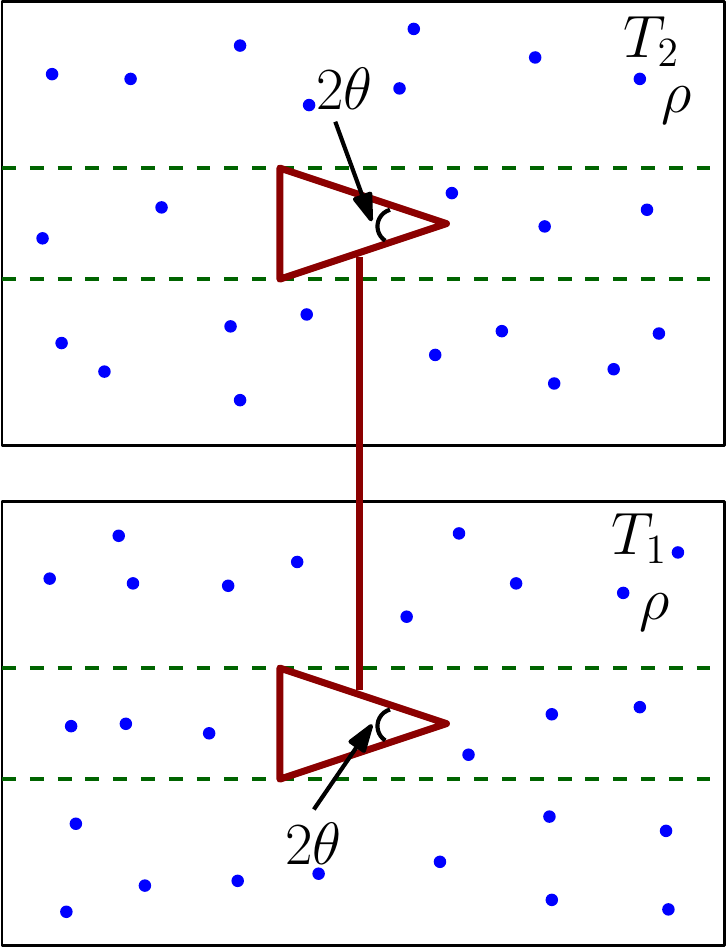}
\caption{Triangula: the triangles can only move horizontally and are connected.  They are in thermal baths at different temperatures.}\label{tir}
\end{figure}  
The speed $V$ of the Triangula depends on the difference in temperatures $T_1, T_2$, and on the apex angle $2\theta$, and to some good approximation for low density reservoirs is given by
\[
V = (1-\sin \theta)\frac{\sqrt{2\pi k_Bm}}{4M}\,(T_1-T_2)\frac{(\sqrt{T_1}-\sqrt{T_2})}{(\sqrt{T_1}+\sqrt{T_2})^2}
\]
for $m$ the mass of the gas particle and $M$ the mass of the triangle (equation 22 in \cite{cvdb} for equal densities).  If we write $T_2 = T_1 + \varepsilon, T_1=T$, that formula becomes in leading order
\[
V \simeq (1-\sin \theta)\frac{\sqrt{2\pi k_Bm}}{4M}\,\frac{\varepsilon
^2}{8T^{3/2}}
\]  
and we see that the speed or current is second order in the temperature difference. That is not attainable with linear response theory around equilibrium. The reason is that the translation current is orthogonal to the heat current (through the rod). Since we are thus in the regime of nonlinear response, that should already tell us that non-dissipative features play a role, cf. \cite{pccp}. As far as we know, nobody has a good heuristics or simple argument to explain that indeed $V>0$.

 \subsubsection{Parrondo game}
 \label{pars}
 
 The following is a paradoxical game invented by Juan Parrondo (1996); see \cite{paro} for more explanations and references.\\
 The state space is $K= \{1,2,3\}$ and the state at time $n$ is $x_n$.  The Markov chain uses a different rule ($A$ or $B$) at even and at odd times $n$.  Alternating, the following two games are played. Game $A$ is fair coin tossing: we simply move $x\rightarrow x\pm 1 \mod 3$ with equal probability at even times. Game
 $B$ is played at odd times and with two biased coins, a good one and a bad one. In
 game $B$, the good coin is tossed when $x_n \in \{1,2\}$ and the bad coin is used each time when $x_n=3$. Winning takes $x_{n+1} = x_n + 1$; losing at time $n$ means
 $x_{n+1} = x_n - 1$, always modulo 3. The transition
 probabilities are then
 \begin{eqnarray}
 \mbox{Prob}[x_{n+1}=x\pm 1|x_n=x] &=& 1/2,\quad \mbox{ when } n \mbox{ is even}
 \nonumber\\
  \mbox{Prob}[x_{n+1} = x + 1|x_n=x] &=& 3/4,\quad \mbox{ when
 } n \mbox{ is odd and } x \neq 3\nonumber\\
  \mbox{Prob}[x_{n+1} = x + 1|x_n=x] &=& 1/10,\quad \mbox{ when
 } n \mbox{ is odd and } x = 3
 \end{eqnarray}
Both games, when played separately at all times are reversible.  For example, for game $B$ (at all times), consider the cycle $3 \rightarrow 1 \rightarrow 2
 \rightarrow 3$ .  Its stationary probability (always for game $B$ alone) is Prob$[3
 \rightarrow 1 \rightarrow 2 \rightarrow 3] = \rho(3)\times 1/10
 \times 3/4 \times 3/4 = 9\rho(3)/160$.  For the reversed
 cycle, the probability Prob$[3 \rightarrow 2 \rightarrow 1
 \rightarrow 3] = \rho(3)\times 9/10 \times 1/4 \times 1/4=
 9\rho(3)/160$ is the same.  The equilibrium distribution for game $B$ is then found to be
 $\rho(1) = 2/13, \rho(2)=6/13$ and $\rho(3) = 5/13$. Obviously then,
 there is no current  when playing game $B$ and clearly,
 the same is trivially verified for game $A$ when tossing with the
 fair coin.  Yet, and here is the paradox, when playing periodically game $B$ after game $A$, a current
 arises.\\
 As in the previous case of the Triangula the very fact that a current arises is not so strange again, but the question is what really decides its direction.  We will show how to solve that question for a continuous time version at low temperature in Section \ref{lT}.

 \subsection{Multiple cycles}\label{mults}
 
 It is not uncommon in nonequilibrium to have multiple cycles in state space along which the dynamics can proceed.  Depending on the cycle  a particular current would go one way or the other, and yet both directions show exactly the same entropy production.  We have illustrated that via an example that models Myosin motion in \cite{physA}.  Here we reduce it to the essential mathematics for a random walker.  Entropy production decides the orientation of a rotational current, but not that of the induced translational current.\\
 \begin{figure}[h]
  \centering
  \includegraphics[width=13 cm]{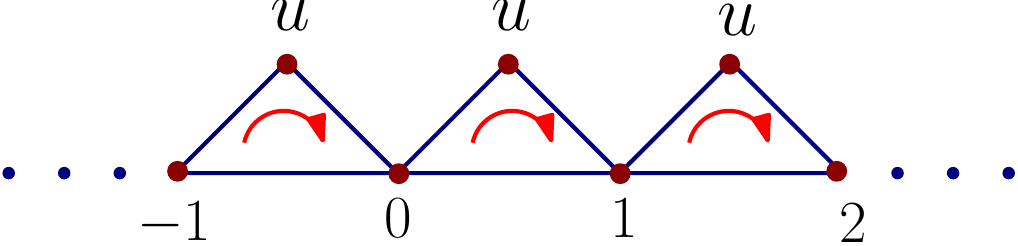}
  \caption{Necklace of three-state cycles with rotational current inducing a horizontal current.}\label{3x}
  \end{figure}
  
  The simplest case is a sequence of triangles where each triangle represents a three state Markov process.  We denote the states by $\{0,u,1\}$, and the transition rates are
 \begin{equation}\label{k3}
 k(1,0) =\varphi\, e^{\varepsilon/2}, \quad k(0,u) = e^{\varepsilon/2}, \quad k(u,1) = e^{\varepsilon/4}
 \end{equation}
 \[
  k(0,1) =\varphi \,e^{-\varepsilon/2}, \quad k(u,0) = 1, \quad k(1,u) = e^{-\varepsilon/4}
  \] 
  for parameters $\varphi,\varepsilon > 0$, see Fig.~\ref{3x}.  The $\varepsilon$, which decides the direction of the rotational current and which is responsible for the breaking of detailed balance, stands for an entropy flux (per $k_B$).  The trajectory $0\rightarrow u\rightarrow 1$ (taking the walker one step to the right) expends an entropy flux $\varepsilon$ (e.g. in the sense of \cite{jmp2000}), as seen from the calculation
  \[
\frac{k(0,1/2) k(1/2,1)}{k(1/2,0)k(1,1/2)} = 
 \frac{e^{\varepsilon/2}e^{\varepsilon/4}}{1\cdot e^{-\varepsilon/4}} = e^{\varepsilon}
\]
but exactly so does the step $0\rightarrow -1$, taking the walker one step to the left:
\[
\frac{k(1,0)}{k(0,1)} = e^{\varepsilon}
\]
There are two ``channels'' to move to the right $0\rightarrow u\rightarrow 1$ and $0\rightarrow 1$, and two ``channels'' 
$0\rightarrow u\rightarrow -1$ and $0\rightarrow -1$ to move to the left. Going right, the system prefers the ``channel''  $0\rightarrow u\rightarrow 1$, and for going left the system prefers the channel $0\rightarrow -1$. In all, there is no entropic preference to go right or left. In other words, the effective bias is also decided by the parameter $\varphi$.  The physical ``translational'' current towards the right is
\[
J = \rho(0)\, [e^{\varepsilon/2} + \varphi\, e^{-\varepsilon/2}]  -  [ \rho(u) + \rho(0) \,\varphi\, e^{\varepsilon/2} ]
\]

where the stationary occupations satisfy $\rho(0) + \rho(u)=1$ and 
\[
\rho(0) \,[e^{\varepsilon/2} + e^{-\varepsilon/4}] = \rho(u)[1 + e^{\varepsilon/4}]
\]
\begin{figure}[h]
  \centering
  \includegraphics[width=14 cm]{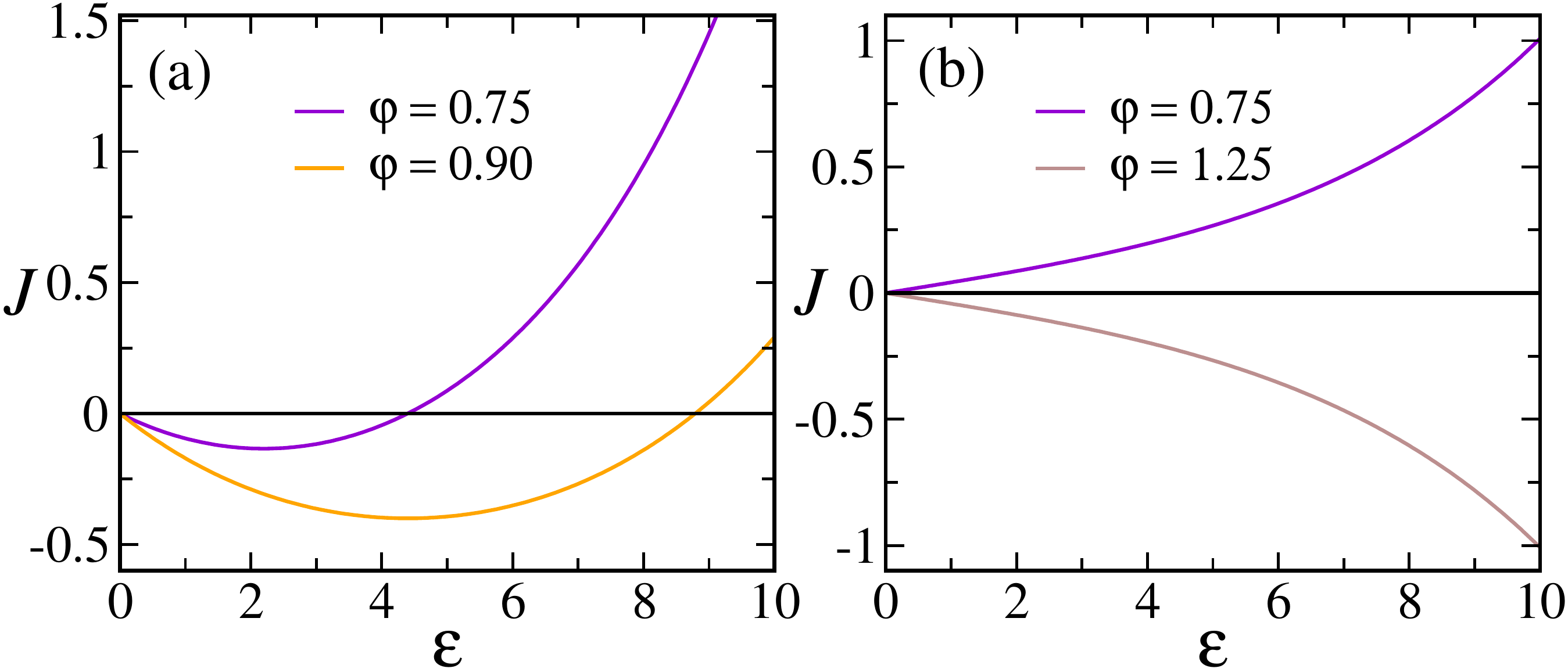}
  \caption{The horizontal current towards the right as a function of $\varepsilon$ corresponding to (a) Fig.~\ref{3x} for  $\varphi=0.65$ (upper) and  $\varphi=0.8$ (lower curve), and (b) Fig.~\ref{4x} for  $\varphi=0.75$ (upper) and  $\varphi=1.25$ (lower curve).}\label{3stphyscur}
  \end{figure}
 
The current $J$ is plotted in Fig.~\ref{3stphyscur}(a), as a function of $\varepsilon$ for two different choices of $\varphi$.  Fixing say $\varepsilon=6$ we see a positive current for $\varphi=0.75$ and a negative current for $\varphi=0.90$.  In other words, the direction of the current is not simply decided. The current diverges like $(1-\varphi)\exp[\varepsilon/4]$ for $\varphi\neq 1$ as $\varepsilon\uparrow \infty$.
  If $\varphi$ is large the current is to the left, and if $\varphi$ is small, the current gets positive.
  For $\varphi> 1/2$ there is a sign-reversal in the current as function of $\varepsilon$.\\

We can add more symmetry in the construction by considering, minimally, a four-state Markov process as elementary unit.
 \begin{figure}[h]
  \centering
  \includegraphics[width=12 cm]{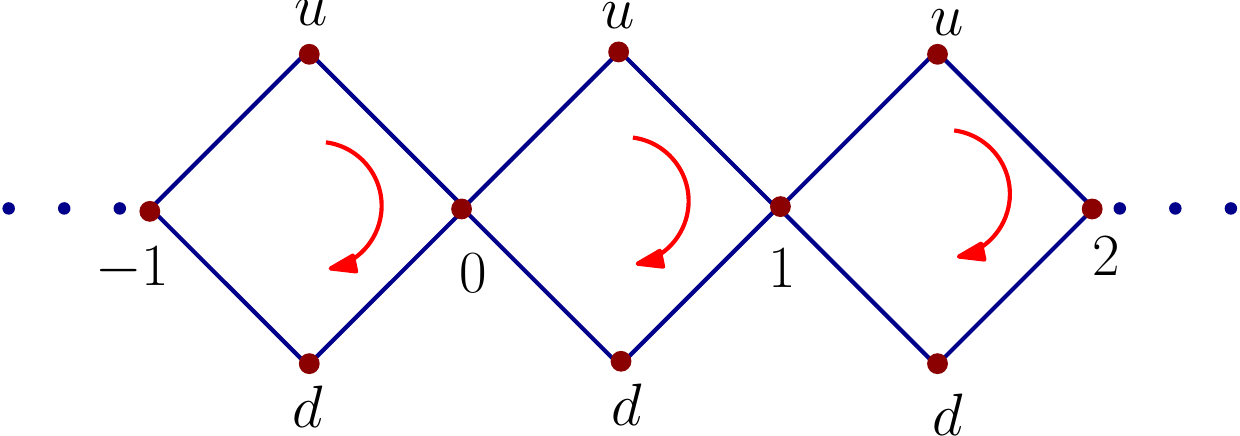}
  \caption{4-state necklace with top versus bottom symmetry for $\varphi=1$.}\label{4x}
  \end{figure}
 Then we look at Fig.~\ref{4x}, denote the states by $\{0,u,1,d\}$ (where $u$ stands for ``up'' and $d$ stands for ``down''), and the transition rates are
\[
 k(1,d) =\varphi \,e^{\varepsilon/4}, \quad k(d,0) = \varphi\, e^{\varepsilon/4}, \quad k(0,u) = e^{\varepsilon/4},\quad k(u,1) = e^{\varepsilon/4}
 \]
 \[
  k(d,1) =\varphi \,e^{-\varepsilon/4}, \quad k(0,d) = \varphi\, e^{-\varepsilon/4}, \quad k(u,0) = e^{-\varepsilon/4},\quad k(1,u) = e^{-\varepsilon/4}
  \] 
 We have again two elementary paths in the opposite direction, which are now $R_1: 0\rightarrow u \rightarrow 1$ and 
$R_2: 0\rightarrow d \rightarrow -1$, for which the entropy fluxes are both equal to $\varepsilon$.  (Of course $R_2$ can be identified with the path $1\rightarrow d \rightarrow 0$.)  It will again be also the ``reactivity'' $\varphi>0$ that decides the direction of the current, see Fig.~\ref{3stphyscur}(b). Or, what starts out as a time-symmetric parameter, turns out to give rise to time-asymmetry.\\

The above scenario has natural realizations, e.g. in the motion of some molecular motors like Myosin V studied in \cite{physA}.  The $\varphi$ then corresponds to the activity of the leading head. Is it because it is lower than that of the trailing head, that the motor moves forward.  The relation between dynamical activity or ``happy feet'' (of Paulo Conte's song) and the direction of current has been anticipated in  La Divina Commedia - Inferno - Canto I, where 
\newpage
Dante writes,
\begin{quotation}
 E come quei che con lena affannata,\\
uscito fuor del pelago a la riva,\\
si volge a l`acqua perigliosa e guata,\\

cos\`{\i} l'animo mio, ch'ancor fuggiva,\\
si volse a retro a rimirar lo passo\\
che non lasci\`o gi\`a mai persona viva.\\

Poi ch'\`ei posato un poco il corpo lasso,\\
ripresi via per la piaggia diserta,\\
s\`{\i} che 'l pi\`e fermo sempre era 'l pi\`u basso.
\end{quotation}
\vspace*{0.5 cm}
In the cartoon Fig.~\ref{4x} one should imagine the top corresponding to lifting the trailing foot and the bottom to lifting the leading foot; $\varphi>1$ corresponding to a more ``active'' trailing foot which easily moves you forward.
As, in the case of Dante leaving the valley and climbing the mountain, the firm or more stable foot was always the lower ($\varphi<1$), it becomes very difficult to go forward, and not to retreat.

 \subsection{Nonequilibrium internal degrees of freedom}\label{stw}
 We are used to think of internal degrees of freedom as an equilibrium reservoir.  When a ball bounces off the ground, it slightly deforms and warms up, indicating a restitution coefficient which is less than one, \cite{hal};  the entropy gets dissipated in these elastic degrees of freedom.  But what if the ball is ``alive'' or ``active,'' or, to put it in less suggestive language, what if the internal degrees of freedom are in steady nonequilibrium.  Can that not produce extra interesting effects?  The problematic case of the Triangula in Section \ref{tra} can be seen as an example.  The two triangles connected by the vertical rod shown in Fig.~\ref{tir} make one extended object which internally is subject to heat conduction (vertical energy current from higher to lower temperature).\\

 \begin{figure}[thb]
    \centering
    \includegraphics[width=8cm]{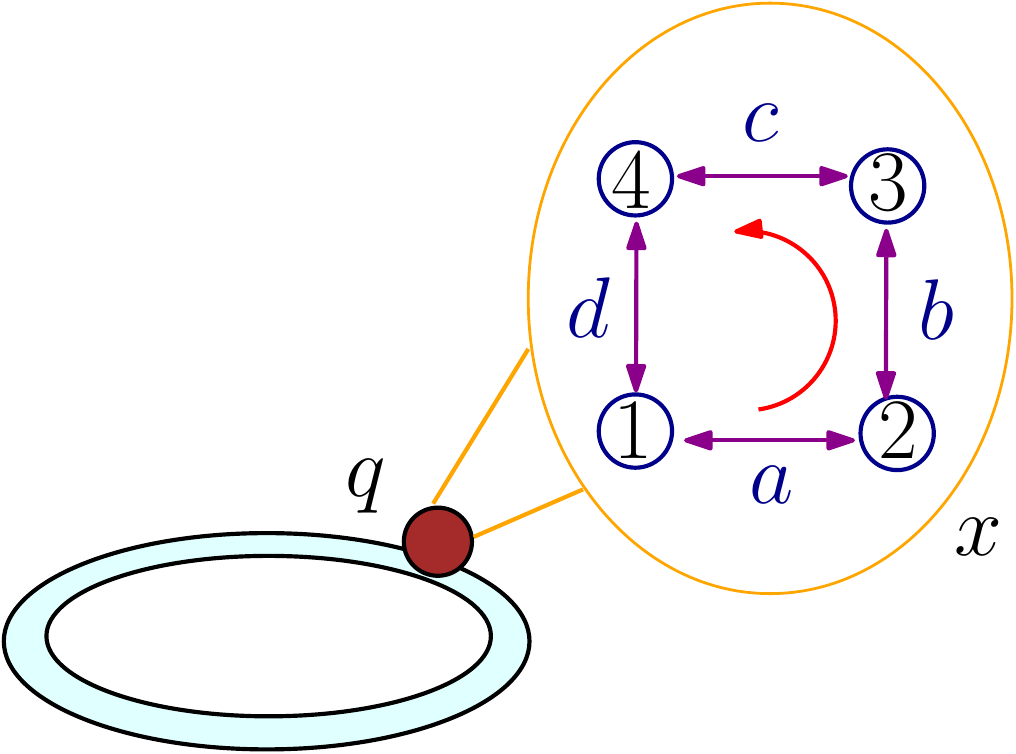}
    \caption{Walker (probe or colloid) on a ring with position $q$ with rotating stomach $x\in \{1,2,3,4\}$. The joint dynamics is specified in \eqref{wal}--\eqref{stom}. The colloid's position is the slow degree of freedom.}
    \label{rond}
  \end{figure}
  
Look now at Fig.~\ref{rond} for greater simplification.  Our object has position $q\in S^1$ on the ring suspended in a thermal bath at inverse temperature $\beta$.  For its dynamics we suppose the overdamped Langevin dynamics
 \begin{equation}\label{wal}
 \gamma\dot{q} = -\frac{\partial}{\partial q} E(x,q) + \sqrt{2\gamma \over \beta}~ \xi_t
 \end{equation}
 in the usual physics notation with $\xi_t$ standard white noise, $\gamma$ the damping coefficient and   
$E(x,q)$ being some interaction energy with an `internal' degree of freedom, here a four state Markov process with $x=1,2,3,4$.  We assume that the $x$ relaxes fast to stationarity compared to the walker where the time-scale is set by $\gamma$, and we take transition rates
 \begin{equation}\label{stom}
 k^q(x,x')= e^{-\frac{\beta}2 [E(x',q)-E(x,q)]} \,\varphi(x,x')\, e^{\frac{1}{2} s(x,x')}
 \end{equation}
The driving or nonequilibrium sits in $s(1,2)= s(2,3)= s(3,4) = s(4,1) =  \beta \varepsilon$ and the symmetric $\varphi(x,x')$  are 
 $\varphi(1,2) = a, \varphi(2, 3)=b, \varphi(3,4)=c, \varphi(4,1) = d$; see the ``stomach'' in Fig. \ref{rond}. Under the hypothesis of infinite time-scale separation the colloid is subject to the mean force
 \begin{equation}\label{fq}
 f(q) = -\sum_x \rho_q(x)\,\frac{\partial}{\partial q}E(x,q) 
 \end{equation}
 which can be calculated exactly from the stationary distribution $\rho_q(x)$ 
 for the internal degree of freedom $x$. When the rotational part of the force
 $f_\text{rot} = \oint f(q)\,\id q$ is non-zero, then the colloid will start moving around the circle. In fact, the steady current $J$, as plotted in Fig.~\ref{4sto} is essentially just given by it. 
 Obviously, there are many parameters, the form of the potential $E(x,q)$ but also the coefficients $a, b, c$ and $d$.  We ask here what determines the sign of that rotational force, which of course determines the direction of the current of the walker.
 
 \begin{figure}[thb]
    \centering
   \includegraphics[width=14cm]{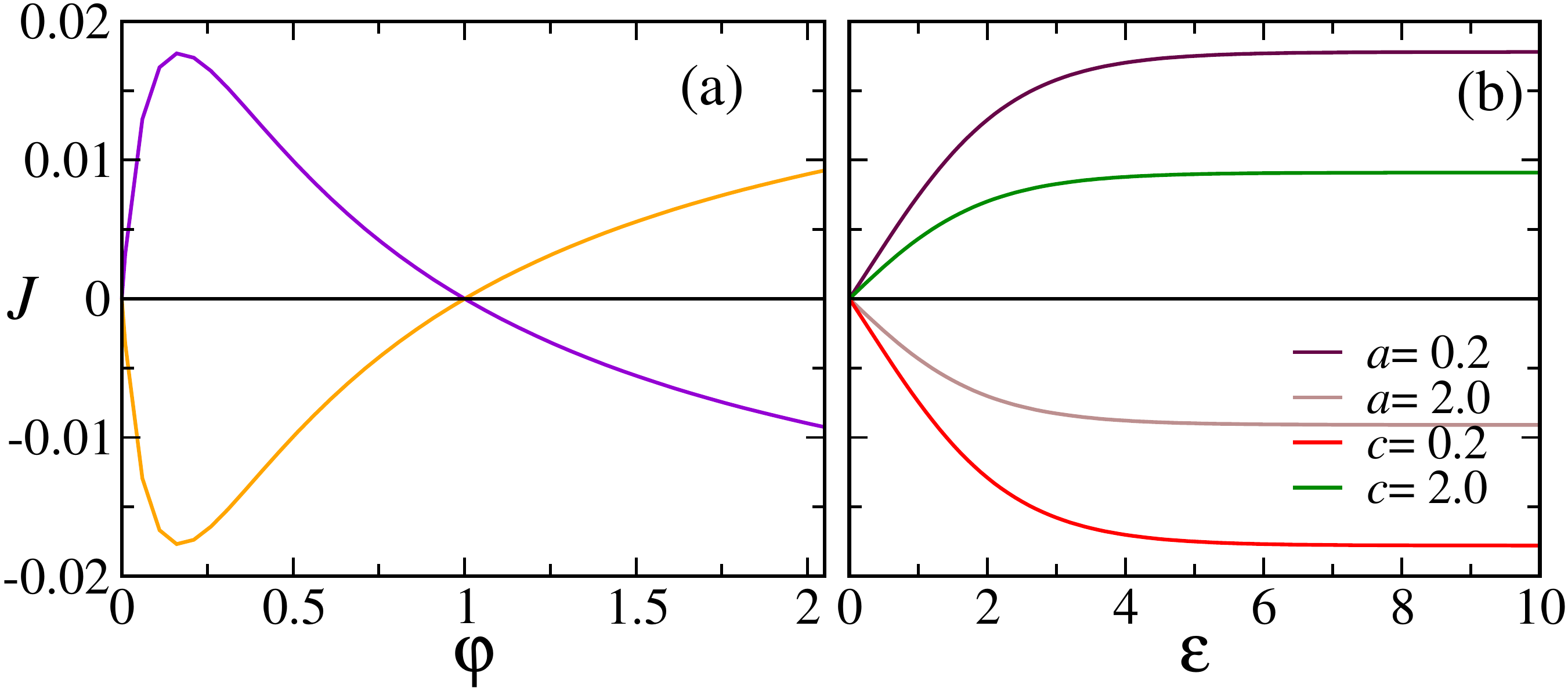}
 \caption{The rotational current $J = f_\text{rot} = \oint f(q)\,\id q$ of  the colloid as it depends on the reactivities ($\varphi$) and the driving $\varepsilon$ of the internal nonequilibrium in Fig.~\ref{rond}, see \eqref{stom}. (a) The direction of the current can change as function of the $a,b,c,d$; we see the current at $\varepsilon=5, \beta=1$ as a function of $\varphi=a$  (in blue), and  as a function of $\varphi=c$ (in orange), while the other reactivities equal 1 when not specified.  (b) The current as function of the driving $\varepsilon$ for various choices of $\varphi(x,x')$. From low to high the curves correspond to $c=0.2, a=2.0, c=2.0$ and $a=0.2$ with again all other reactivities fixed to the value 1 when not specified.}
    \label{4sto}
  \end{figure}
  
We observe here that we can get a sign-reversal of the current by varying solely the kinetic factors $a,b,c,d$.  More specifically we consider the energy function $E(x,q)$ for $q\in [-1,1]$, $E(2,q)=E(4,q)=0$ and

\begin{eqnarray*}
E(1,q)  = \left \{ 
\begin{split}
\frac 47(1+q) \quad \text{for}\quad  q\leq 3/4,\\
4(1-q) \quad \text{for}\quad  q\geq 3/4
\end{split}
\right. \quad \text{and} \quad
E(3,q)= \left \{ 
\begin{split}
4(1+q) \quad \text{for}\quad  q\leq -3/4,\\
\frac 47 (1-q) \quad \text{for}\quad  q\geq -3/4 
\end{split}
\right.
\end{eqnarray*}

In Fig.~\ref{4sto} (a) we see the rotational current $J=f_\text{rot}$ as a function of $c$ for $a=b=d=1$ (first negative then positive) and as a function of $a$ for $b=c=d=1$ (first positive then negative), both at driving $\varepsilon =5$, and $\beta=1$. The same is represented in the right panel but now as a function of the driving $\varepsilon$.  We clearly get information about the time-symmetric part in the transition rates \eqref{stom} from coupling that process $x_t$ to the position $q_t$ in \eqref{wal} of the walker and measuring its induced current.

 \subsection{Wrong direction!}
 
The response to an external field can be negative.  It is then the case that by pushing harder the particle gets slower.  It could even happen, that by pushing in one direction the particle moves in the opposite direction making negative absolute conductivity.  In \cite{chrvb} one uses memory to achieve that result, but one gets it also from considering the Markov models of Section \ref{mults}.\\

 \begin{figure}[h!]
  \centering
 \includegraphics[width=8cm]{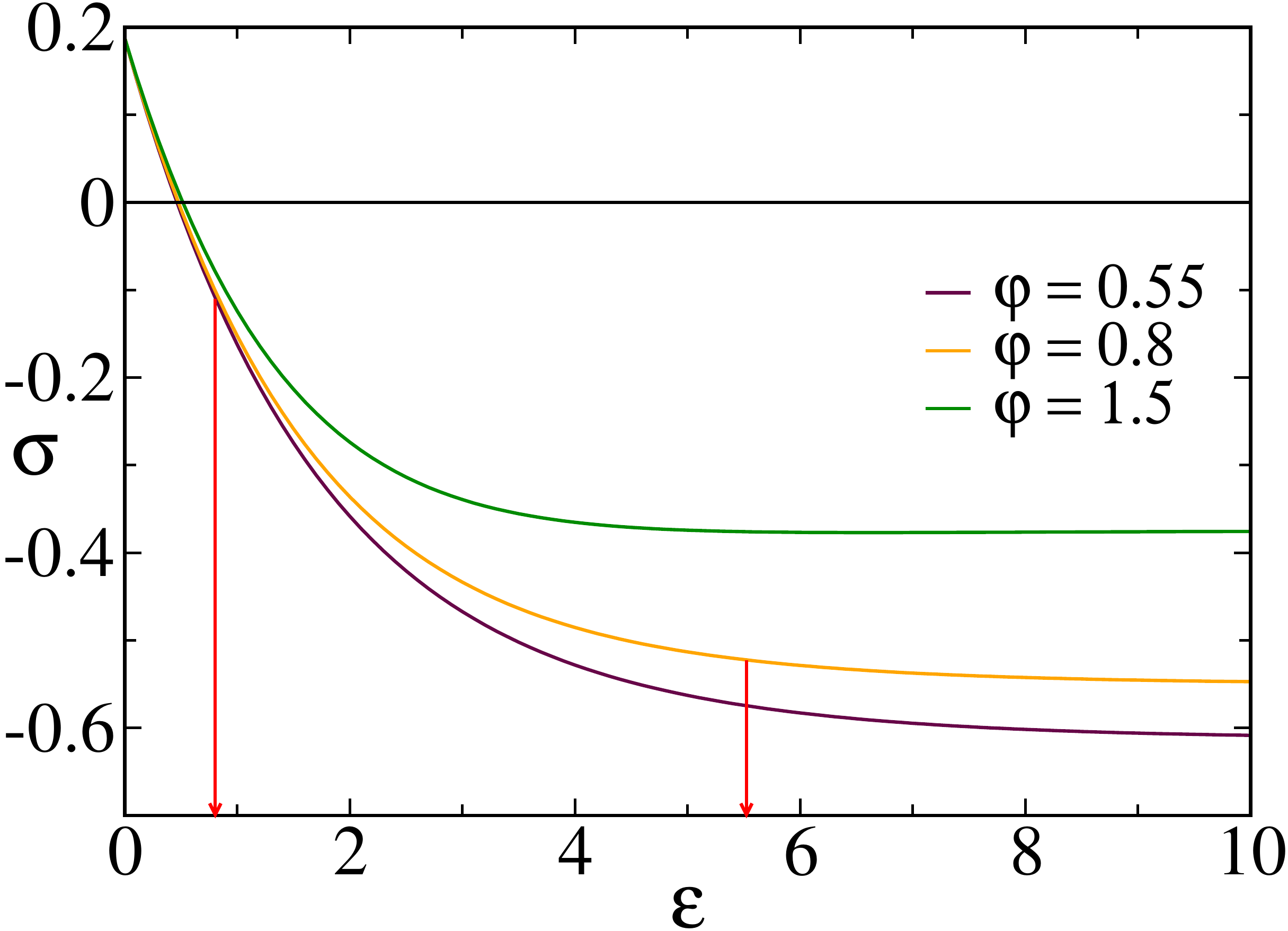}
  \caption{ The conductivity $\sigma$ \eqref{sce}, rescaled with a factor $\exp[-\varepsilon/4]$, as function of $\varepsilon$ for $\varphi=0.55$ (lowest curve) having its stalling point at $\varepsilon^* \simeq 0.80$, for $\varphi=0.8$ with stalling point at $\varepsilon^*\simeq 5.54$ and $\varphi=1.5$ (upper curve). There is negative conductivity $\sigma<0$, including at stalling points where the pushing makes the walker go back instead of forward.}
   \label{conc}
 \end{figure}
 
Consider again the set-up of Fig.~\ref{3x}, and the result \eqref{3stphyscur}(a) for the horizontal current in the positive direction.  There are possible stalling points $(\varepsilon^*,\varphi*)$ where that current vanishes.  Taking these values or, more generally, fixing arbitrary $(\varepsilon,\varphi)$ we perturb the rates \eqref{k3} in the following way:
\[
 k_E(1,0) =[\varphi+E]\, e^{\varepsilon/2}, \quad k_E(0,u) = e^{[\varepsilon+E]/2}, \quad k_E(u,1) = e^{[\varepsilon+E]/4}
 \]
 \[
  k_E(0,1) =[\varphi+E] \,e^{-\varepsilon/2}, \quad k_E(u,0) = 1, \quad k_E(1,u) = e^{-\varepsilon/4}
  \] 
pushing a bit harder with $E>0$ in the upper channel (only) and changing the time-symmetric coefficient $\varphi\rightarrow \varphi + E$ also.  We get a new value of the horizontal current $J_E$ and we can ask how it changes, that is to find the conductivity 
\begin{equation}\label{sce}
\sigma = \left.\frac{\id J_E}{\id E} \right|_{E=0}
\end{equation}

 We see in Fig.~\ref{conc} that $\sigma$ gets negative for large enough values of $\varepsilon$, including at stalling values.  Thus, there, the current actually goes backward while pushing forward. 
 
 \section{Low temperature analysis}\label{lT}
 
We consider here a continuous time version of the  Parrondo game of Section \ref{pars} with random flipping between a flat potential  and a nontrivial energy landscape.  It gives an approach to the problems of Sections \ref{pars}--\ref{mults} by considering low temperature asymptotics.
In particular we use the Freidlin-Wentzel theory of \cite{lowT} to obtain an expression for the low temperature ratchet current.  Its direction is not determined by entropic considerations (only) but involves the reactivities.\\

Look at Fig.\ref{grafsa}.  States of a continuous time Markov process are on two rings, each having $N>2$ of states,  denoted by $x= (i,n)$ where $i \in\{ 1=N+1,2,\ldots,N\}$ and $n=0,1$.

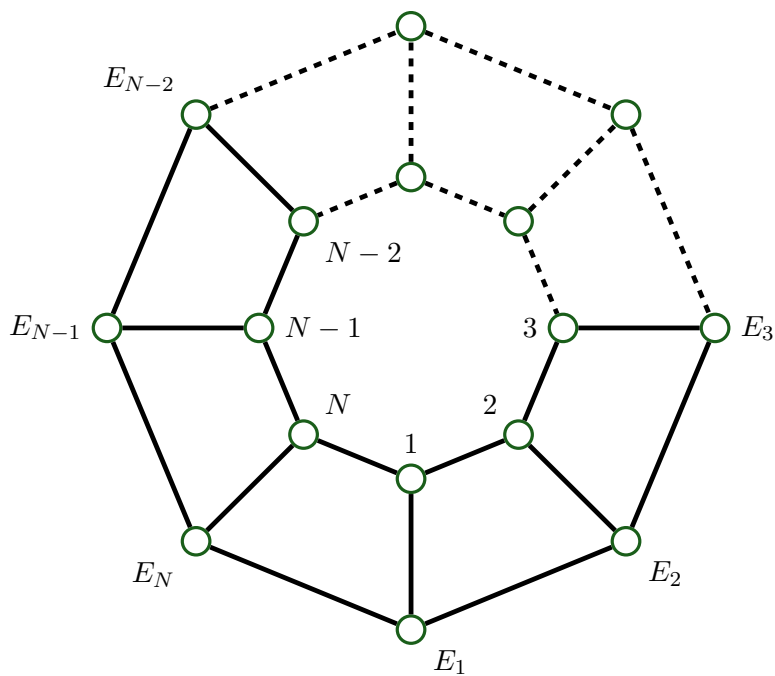
\begin{figure}[h]
\centerline{
\begin{tikzpicture}
\node [green!30!black!90,line width=1.2 pt,draw, shape= circle,label=90:$1$ ](x1) at (270:2)  {$ $};
\node [green!30!black!90,line width=1.2 pt,draw, shape= circle,label=135:$2$] (x2) at (315 :2) {$ $};
\node  [green!30!black!90,line width=1.2 pt,draw, shape= circle,label=180:$3$ ](x3) at (0:2) {$ $};
\node [green!30!black!90,line width=1.2 pt,draw, shape= circle ,label=45:$N$ ] (xN) at (225:2) {$ $};
\node  [green!30!black!90,line width=1.2 pt,draw, shape= circle,label=0:$N-1$ ](xN1) at (180:2) {$ $};
\node  [green!30!black!90,line width=1.2 pt,draw, shape= circle,label=315:$N-2$ ](xN2) at (135:2) {$ $};
\node  [green!30!black!90,line width=1.2 pt,draw, shape= circle ](xDOT1)  at (90:2) {};
\node [green!30!black!90,line width=1.2 pt,draw, shape= circle ] (xDOT2)  at (45:2) {};
\node [green!30!black!90,line width=1.2 pt,draw, shape= circle,label=320:$E_1$] (E1) at (270:4) {$ $};
\node  [green!30!black!90,line width=1.2 pt,draw, shape= circle,label=320:$E_2$ ](E2) at (315 :4) {$ $};
\node  [green!30!black!90,line width=1.2 pt,draw, shape= circle,label=0:$E_3$ ](E3) at (0:4) {$ $};
\node [green!30!black!90,line width=1.2 pt,draw, shape= circle,label=225:$E_N$ ]  (EN) at (225:4) {$ $};
\node [green!30!black!90,line width=1.2 pt,draw, shape= circle,label=180:$E_{N-1}$  ] (EN1) at (180:4) {$ $};
\node [green!30!black!90,line width=1.2 pt,draw, shape= circle,label=135:$E_{N-2}$  ] (EN2) at (135:4) {$ $};
\node  [green!30!black!90,line width=1.2 pt,draw, shape= circle ](DOT1)  at (90:4) {};
\node  [green!30!black!90,line width=1.2 pt,draw, shape= circle ](DOT2)  at (45:4) {};
\foreach \from/\to in {E1/E2, E2/E3, E1/EN, EN2/EN1}
\draw [ thick,line width=1.8 pt] (\from) -- (\to);
\draw[dashed, thick,line width=1.8 pt] (EN2)-- (DOT1);
\draw[dashed, thick,line width=1.8 pt] (DOT2)--(E3) ;
\draw[dashed, thick,line width=1.8 pt ] (DOT1)--(DOT2)  ;
\draw[dashed, thick,line width=1.8 pt] (xDOT1)--(DOT1)  ;
\draw[dashed, thick,line width=1.8 pt ] (xDOT2)--(DOT2)  ;
\draw[dashed, thick,line width=1.8 pt ] (xDOT1) -- (xN2);
\draw[dashed, thick,line width=1.8 pt ] (xDOT1) -- (xDOT2);
\draw[dashed, thick,line width=1.8 pt ] (xDOT2) -- (x3);
\draw[ thick,line width=1.8 pt] (x1) -- (E1);
\draw[ thick,line width=1.8 pt] (EN1) -- (EN);
\foreach \from/\to in {x2/E2,x3/E3,xN/EN,xN1/EN1,xN2/EN2}
\draw[thick,line width=1.8 pt] (\from) -- (\to);
\foreach \from/\to in {x1/x2,x2/x3,xN2/xN1,xN/xN1,x1/xN}
\draw[thick,line width=1.8 pt] (\from) -- (\to);
\end{tikzpicture}
}
\caption{Continuous time Parrondo game.}
\label{grafsa}\end{figure}

On the outer ring ($n=0$) energies  $E_1<\ldots < E_N$ are associated to the states and transition rates 
are thermal,
 \begin{equation}\label{rate1}
k((i,0),(i+1,0))=e^{\frac{\beta}{2} (E_i-E_{i+1})}, \quad k((i+1,0),(i,0))=e^{\frac{\beta}{2} (E_{i+1}-E_{i})}
\end{equation}
for inverse temperature $\beta$.
The inner ring ($n=1$) corresponds to a walker in a flat potential landscape so that
 \begin{equation}\label{rate2}
k((i,1),(i+1,1))=k((i+1,1),(i,1))=1
\end{equation}
The random flipping between the two potentials is realized by moves between the rings, at transition rates $k((i,n),(i,1-n))=a$ for some $a>0$. There is no explicit driving except that for $a=0$ there is detailed balance of course and for very strong coupling $a\gg 1$, the model is effectively running on a single ring.  In the limit $a\uparrow \infty$ there is again detailed balance with inverse temperature $\beta/2$.\\
 The question for the nonequilibrium situation is in what sense the walker will typically move, either clockwise of counter clockwise.  Again indeed, the direction of that current, we call it now the ratchet current, is not decided by the positivity of the entropy production.  Consider for example Fig.~\ref{twotraj} where two trajectories  $\omega_1 = ((N,0),(N-1,0),\ldots,(1,0), (1,1),(N,1),(N,0))$ and $\omega_2 = ((N,0),(1,0),(1,1),(2,1),\ldots,(N,1),(N,0))$ are depicted 
 \begin{figure}[h]\label{twotraj}
\centerline{
\begin{tikzpicture}
\node [green!30!black!90,line width=1.2 pt,draw, shape= circle,label=90:$1$ ](x1) at (270:2)  {$ $};
\node [green!30!black!90,line width=1.2 pt,draw, shape= circle,label=135:$2$] (x2) at (315 :2) {$ $};
\node  [green!30!black!90,line width=1.2 pt,draw, shape= circle,label=180:$3$ ](x3) at (0:2) {$ $};
\node [green!30!black!90,line width=1.2 pt,draw, shape= circle ,label=45:$N$ ] (xN) at (225:2) {$ $};
\node  [green!30!black!90,line width=1.2 pt,draw, shape= circle,label=0:$N-1$ ](xN1) at (180:2) {$ $};
\node  [green!30!black!90,line width=1.2 pt,draw, shape= circle,label=315:$N-2$ ](xN2) at (135:2) {$ $};
\node  [green!30!black!90,line width=1.2 pt,draw, shape= circle ](xDOT1)  at (90:2) {};
\node [green!30!black!90,line width=1.2 pt,draw, shape= circle ] (xDOT2)  at (45:2) {};
\node [green!30!black!90,line width=1.2 pt,draw, shape= circle,label=320:$E_1$] (E1) at (270:4) {$ $};
\node  [green!30!black!90,line width=1.2 pt,draw, shape= circle,label=320:$E_2$ ](E2) at (315 :4) {$ $};
\node  [green!30!black!90,line width=1.2 pt,draw, shape= circle,label=0:$E_3$ ](E3) at (0:4) {$ $};
\node [green!30!black!90,line width=1.2 pt,draw, shape=circle,label=225:$E_N$ ]  (EN) at (225:4) {$ $};
\node [green!30!black!90,line width=1.2 pt,draw, shape= circle,label=180:$E_{N-1}$  ] (EN1) at (180:4) {$ $};
\node [green!30!black!90,line width=1.2 pt,draw, shape= circle,label=135:$E_{N-2}$  ] (EN2) at (135:4) {$ $};
\node  [green!30!black!90,line width=1.2 pt,draw, shape= circle ](DOT1)  at (90:4) {};
\node  [green!30!black!90,line width=1.2 pt,draw, shape= circle ](DOT2)  at (45:4) {};
\foreach \from/\to in {E2/E1, E3/E2,EN1/EN2}	
\draw [blue,-> , thick, line width=1.8pt] (\from) -- (\to);
\draw[blue,dashed,-> , thick, line width=1.8pt] (EN2)-- (DOT1);
\draw[blue,dashed,-> , thick, line width=1.8pt] (DOT2)--(E3) ;
\draw[blue,dashed,->, thick, line width=1.8pt ] (DOT1)--(DOT2)  ;
\draw[dashed,-, thick, line width=1.8pt ] (xDOT1)--(DOT1)  ;
\draw[dashed,-, thick, line width=1.8pt ] (xDOT2)--(DOT2)  ;
\draw[red,dashed, thick, -> , line width=1.8pt] (xDOT1) -- (xN2);
\draw[red,dashed, thick, -> , line width=1.8pt] (xDOT2) -- (xDOT1);
\draw[red,dashed, thick, -> , line width=1.8pt] (x3) -- (xDOT2);
\draw[red, thick, -> , line width=1.8pt] (EN) -- (E1);
\draw[blue, thick, ->, line width=1.8pt ] (EN) -- (EN1);
\draw[thick,blue,-> , line width=1.8pt ] (x1) -- (xN);
\foreach \from/\to in {x2/E2,x3/E3,xN1/EN1,xN2/EN2}
\draw[-, thick, line width=1.8pt] (\from) -- (\to);
\foreach \from/\to in {x1/x2,x2/x3,xN2/xN1,xN1/xN}
\draw[->, thick,red, line width=1.8pt] (\from) -- (\to);
  \path[blue,thick,->, line width=1.8pt,every node/.style={font=\sffamily\small}]
(E1) edge [bend right] (x1);
  \path[red,thick,->, line width=1.8pt,every node/.style={font=\sffamily\small}]
(E1) edge [bend left] (x1);
  \path[blue,thick,->, line width=1.8pt,every node/.style={font=\sffamily\small}]
(xN) edge [bend right] (EN);
  \path[red,thick,->, line width=1.8pt,every node/.style={font=\sffamily\small}]
(xN) edge [bend left] (EN);
\end{tikzpicture}
}
\caption{Trajectories $\omega_1$ (clockwise, blue) and $\omega_2$ (counter-clockwise, red) with the same entropy flux, yet in opposite directions.}
\end{figure}
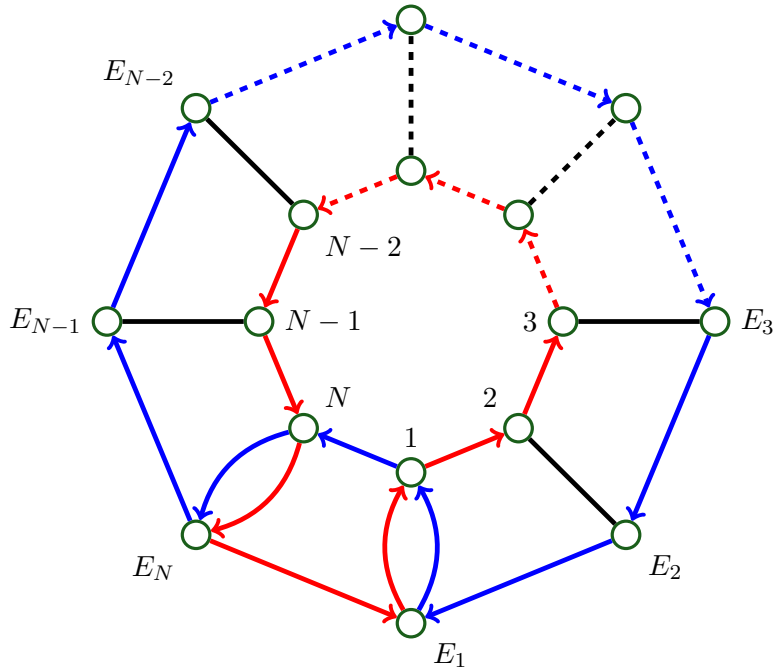
that wind in opposite directions; yet, their entropy fluxes are exactly identical, equal to $s(\omega_1)=s(\omega_2) = \beta(E_N - E_1) > 0$.\\

The stationary ratchet current $J_R$ in the clockwise direction is 
 \begin{equation}
J_{R}=j((i+1,0),(i,0))+ j((i+1,1),(i,1))
\end{equation}
with $j(x,y) = k(x,y)\rho(x) -k(y,x)\rho(y)$ where $\rho$ is the stationary probability law. It is  the current over both rings at the same time, and of course that current also depends on  the size $N$,  on the energies and on temperature.  We will look at the case $a=1$ but at low temperatures so that the transitions $(i,0)\rightarrow (i+1,0)$ are exponentially damped. The following combines proofs in \cite{lowT} and in \cite{louis} to show that the ratchet current is clockwise and saturates.
\begin{proposition}
\[
J_R = 0 \text{ for } N=3,\quad J_R > 0 \text{ for all } N>3 \text{ and } \lim_{N\uparrow \infty} J_R= \frac{1}{2} -\frac 1{\sqrt{5}}
\]
\end{proposition}
\begin{proof}
   Consider the set ${\cal D} := \{(1,0),(i,1),i=1,\ldots,N\}$ and let $M(x)$ be the number of in-spanning trees in the digraph obtained from Fig.\ref{grafsa} by keeping only the oriented bonds $(v,w)$ where $w$ is one of the most likely successor of $v$.  From \cite{lowT} we learn that at low temperatures, $\rho(x) \propto |M(x)|$ for $x\in {\cal D}$, and $\rho(y) \simeq |M(y)|\,e^{\beta \Gamma(y)}/{\cal Z}$, with some $\Gamma(y)<0$ for $y\notin \cal D$. By the Matrix-Tree Theorem; see e.g. \cite{Tutte}, we need the Laplacian matrix $L$ on the digraph $K^D$ and we erase the row and the  column corresponding to vertex $x$ to obtain the matrix $L_x$.  Then,
 \begin{equation}
|M(x)| = \det L_x
\end{equation}
The Laplacian of the digraph $K^D$ has a rather simple structure:
 \begin{eqnarray}\nonumber
L = \bordermatrix{ 
 ~      &  (1,0)       & (2,0) & \hdots  &\hdots&(N,0)&(1,1)&  (2,1) &\hdots	& \hdots&(N,1) \\	\nonumber
(1,0)   &   1          &       &         &      &     &-1   &        &    	  &       &       \\	\nonumber
(2,0)   &  -1  &  1    &        &       &       &    &       &		  &    	  &	   \\   \nonumber
\vdots&      & \ddots& \ddots &       &       &    &       &		  & 	  &    \\   \nonumber
\vdots&      &       & -1     &1      &       &    &       &		  & 	  &    \\   \nonumber
(N,0)   &   -1  &       &        &    0  & 1     &    &       &		  & 	  &    \\\hline \nonumber
(1,1)   &   -1 &       &        &       &       & 3  & -1    &		  & 	  &   -1 \\	\nonumber
(2,1)   &      & -1    &        &       &       & -1 & 3     & \ddots   &	      &    \\   \nonumber 
\vdots&      &       & \ddots &       &       &    & \ddots& \ddots   & \ddots&    \\   \nonumber 
\vdots&      &       &        & \ddots&       &    &       & \ddots   & \ddots& -1 \\   \nonumber 
(N,1)   &      &       &        &       &-1     & -1   &       &          & -1    & 3\\}    \nonumber
\end{eqnarray}
The state for which the number of in-trees becomes maximal is $(1,0)$: there are more combinations to form an in-tree to $(1,0)$ than to any other state $(i,1)$ on the inner ring.\\
To compute the ratchet current we take $x=(1,1)$ for which
 $\rho(1,1)\simeq\frac{1}{\mathcal{Z}}A((1,1))$.  Then, 
 \begin{equation}\nonumber
j((2,1),(1,1)) \simeq \frac{1}{\mathcal{Z}}\left(A((2,1))-A((1,1))\right)
\end{equation} 
Moreover,
 \begin{equation}\nonumber
j((2,0),(1,0)) \simeq \frac{A((2,0))}{\mathcal{Z}}
\end{equation}
As  a consequence,
 \begin{equation}\nonumber
J_{R}\simeq\frac{1}{\mathcal{Z}}\left( \det L_{(2,1)}+\det L_{(2,0)} - \det L_{(1,1)}\right)
\end{equation}
Furthermore, by inspecting the Laplacian $L$, one finds that
\begin{enumerate}
\item $\det L_{(2,0)}=2\det B_{N-1} -3\det B_{N-2}-3$,
\item $\det L_{(1,1)}=\det B_{N-1}$,
\item $\det L_{(2,1)}=\det B_{N-2}+1$
\end{enumerate}
with 
 \begin{equation}
B_N =\left( \begin{array}{cccc}
   3  & -1    &        &              \\	
  -1  & 3    & \ddots &         \\
      & \ddots& \ddots & -1         \\
      &       & -1 & 3        \\ 
 \end{array}\right)
\end{equation}
$B_N$ satisfies the recursion relation $\det B_N=3\det B_{N-1} - \det B_{N-2}$, where $\det B_2=8$ and $\det B_1=3$. Hence, by solving the recurrence we get
\[
\det B_N = \frac{5-3\sqrt{5}}{10}\,\large(\frac{3-\sqrt{5}}{2}\large)^N + \frac{5+3\sqrt{5}}{10}\,\large(\frac{3+\sqrt{5}}{2}\large)^N
\]
to be used in
\[
J_{R} \simeq \frac{\det B_{N-1}-2\det B_{N-2}-2}{\mathcal{Z}}
\]
which already proves that $J_{R}>0$, $\forall N\geq 4$, and $J_{R}=0$ when $N=3$; the direction is clockwise.  For the $N-$asymptotics we also need the normalization $\cal Z$.  In fact, ${\cal Z} \simeq \sum_{x\in \cal D}|M(x)| = \sum_{x\in \cal D}\det L_x$.  In \cite{louis} it is shown that 
\[
{\cal Z} \simeq 
2\,\large(\frac{3-\sqrt{5}}{2}\large)^N + 2\,\large(\frac{3-\sqrt{5}}{2}\large)^N - 4
\]
which concludes the proof by a simple computation. 
  \end{proof}
Note that adding particles and interactions we can get direction-reversal of currents as we had it in the previous sections \ref{mults}--\ref{stw}.  An experimentally accessible example is described in \cite{mol}. That constitutes another big challenge in the discussion of the direction of currents that has not been touched here, how density and interactions can modify it.

\section{Conclusion}

To discover what decides the direction of a current under steady nonequilibrium conditions is a major challenge of statistical mechanics.  In the present paper we have seen that many effects are possible, not in the least from the variation of time-symmetric parameters in transition rates defining the process.
That dependence on non-dissipative aspects provides thus a method to obtain kinetic parameters from measuring the direction of the current.\\

\noindent {\bf Acknowledgment}:  I thank Urna Basu for many discussions and help.


\end{document}